\documentclass{eptcs-modified}

\usepackage{amsmath}
\usepackage{amssymb}
\usepackage{url}
\usepackage{amsthm}
\usepackage{graphicx}
\usepackage[usenames,dvipsnames,svgnames,table]{xcolor}

\begin{document}

\theoremstyle{definition}
\newtheorem{theorem}{Theorem}
\newtheorem{definition}[theorem]{Definition}
\newtheorem{problem}[theorem]{Problem}
\newtheorem{assumption}[theorem]{Assumption}
\newtheorem{corollary}[theorem]{Corollary}
\newtheorem{proposition}[theorem]{Proposition}
\newtheorem{example}[theorem]{Example}
\newtheorem{lemma}[theorem]{Lemma}
\newtheorem{observation}[theorem]{Observation}
\newtheorem{fact}[theorem]{Fact}
\newtheorem{question}[theorem]{Open Question}
\newtheorem{conjecture}[theorem]{Conjecture}
\newtheorem{addendum}[theorem]{Addendum}
\newcommand{\uint}{{[0, 1]}}
\newcommand{\Cantor}{{\{0,1\}^\mathbb{N}}}
\newcommand{\name}[1]{\textsc{#1}}
\newcommand{\id}{\textrm{id}}
\newcommand{\lpo}{\textrm{LPO}}
\newcommand{\llpo}{\textrm{LLPO}}
\newcommand{\dom}{\operatorname{dom}}
\newcommand{\Dom}{\operatorname{Dom}}
\newcommand{\codom}{\operatorname{CDom}}
\newcommand{\Baire}{{\mathbb{N}^\mathbb{N}}}
\newcommand{\hide}[1]{}
\newcommand{\mto}{\rightrightarrows}
\newcommand{\Sierp}{Sierpi\'nski }
\newcommand{\BC}{\mathcal{B}}
\newcommand{\C}{\textrm{C}}
\newcommand{\aouc}{\textrm{AoUC}_{\uint}}
\newcommand\tboldsymbol[1]{%
\protect\raisebox{0pt}[0pt][0pt]{%
$\underset{\widetilde{}}{\boldsymbol{#1}}$}\mbox{\hskip 1pt}}

\newcommand{\bcode}{{\rm BC}}
\newcommand{\bcodefun}{\pi}
\newcommand{\bolds}{\tboldsymbol{\Sigma}}
\newcommand{\boldp}{\tboldsymbol{\Pi}}
\newcommand{\boldd}{\tboldsymbol{\Delta}}
\newcommand{\boldg}{\tboldsymbol{\Gamma}}

\newcommand{\leqW}{\leq_{\textrm{W}}}
\newcommand{\nleqW}{\nleq_{\textrm{W}}}
\newcommand{\leW}{<_{\textrm{W}}}
\newcommand{\equivW}{\equiv_{\textrm{W}}}
\newcommand{\geqW}{\geq_{\textrm{W}}}
\newcommand{\pipeW}{|_{\textrm{W}}}

\newcommand{\BRoot}{\mathrm{BRoot}}
\newcommand{\Nash}{\mathrm{Nash}}

\newcommand{\proj}{\textrm{proj}}
\newcommand{\psc}{\textrm{psc}}
\newcommand{\F}{\mathcal{F}}
\newcommand{\R}{\mathbb{R}}
\newcommand{\BMRoot}{\mathrm{BMRoot}}

\newcommand{\tc}[1]{\textcolor{purple}{#1}}

\title{The Weihrauch degree of finding Nash equilibria in multiplayer games}

\author{
Tonicha Crook\thanks{This work is supported by the UKRI AIMLAC CDT, cdt-aimlac.org, grant no. EP/S023992/1.} 
~\& Arno Pauly
\institute{Department of Computer Science\\Swansea University\\Swansea, UK}
\email{t.m.crook15@outlook.com}
\email{Arno.M.Pauly@gmail.com}
}

\def\titlerunning{Finding Nash equilibria in multiplayer games}
\def\authorrunning{T.~Crook \& A.~Pauly}
\maketitle

\begin{abstract}
Is there an algorithm that takes a game in normal form as input,
and outputs a Nash equilibrium? If the payoffs are integers, the
answer is yes, and a lot of work has been done in its computational
complexity. If the payoffs are permitted to be real numbers, the
answer is no, for continuity reasons. It is worthwhile to
investigate the precise degree of non-computability (the Weihrauch
degree), since knowing the degree entails what other approaches are
available (eg, is there a randomized algorithm with positive
success change?). The two-player case has already been fully
classified, but the multiplayer case remains open and is addressed
here. Our approach involves classifying the degree of finding roots
of polynomials, and lifting this to systems of polynomial
inequalities via cylindrical algebraic decomposition.

ACM classification: Theory of computation--Computability; Mathematics of computing--Topology; Mathematics of computing--Nonlinear equations
\end{abstract}

%\begin{keywords}
% computable metric space, recursive presentation, recursive Polish space, effectively Borel measurable, Cauchy completion
%\end{keywords}

\section{Introduction}
Is there an algorithm that reads games in strategic form and outputs some Nash equilibrium? This question is not only relevant for practical applications of Nash equilibria, whether in economics or computer science, but it also plays a central role in justifying Nash equilibrium as the outcome of rational behavior. If an agent can discern their own strategy in a Nash equilibrium (in order to follow it), then all agents together ought to be able to compute a Nash equilibrium.

If the payoffs in our games are given as integers, the existence of algorithms to find Nash equilibria is readily verified. Here the decisive question is how efficient these algorithms can be. For two-player games, the problem is $\mathrm{PPAD}$-complete \cite{denge}, whereas the multiplayer variant is complete for $\mathrm{FIXP}$ \cite{etessami}\footnote{The $\mathrm{PPAD}$-completeness result from \cite{daskalakis} is for $\varepsilon$-Nash equilibria, not for actual Nash equilibria.}. These complexity classifications are still a challenge for justifying Nash equilibrium as a solution concept, since they are widely believed to be incompatible with the existence of efficient algorithms.

Our focus here is on payoffs given as real numbers. Essentially, this means that our algorithm has access to arbitrarily good approximations for the payoffs. However, an algorithm cannot confirm that two real inputs are equal -- and it cannot even pick a true case between \emph{the first input is not smaller} and \emph{the second input is not smaller}. This is the non-constructive principle $\mathrm{LLPO}$, which is easily seen to correspond to finding Nash equilibria in single-player games with just two options. We should not stop with this negative answer to our initial question, but instead, explore \emph{how non-computable} the task of finding Nash equilibria is. As usual, this means identifying the degree of the problem for a suitable notion of reducibility. Here, this notion is Weihrauch reducibility.

Besides satisfying our curiosity, classifying the Weihrauch degree of finding Nash equilibria lets us draw some interesting conclusions. For example, there is a Las Vegas algorithm to compute Nash equilibria, but we cannot provide any lower bound for its success rate. We will discuss these consequences further in Section \ref{sec:consequences}. For games with one or two players, a complete classification has already been obtained in \cite{paulyincomputabilitynashequilibria}, but the situation for multiplayer games remained open and will be addressed here.

We use the well-known ``algorithm'' called Cylindrical Algebraic Decomposition (CAD) with a few modifications to reach our results. The modifications are necessary because in its original form, CAD assumes the equality of coefficients to be decidable. We explore the computable content of CAD by investigating each aspect of the algorithm to what extent they are computable when working with real numbers. The obstacles can be overcome by moving to suitable over-approximations.

\section{Constructivism in Game Theory \& Bounded Rationality}
\label{sec:context}
The study of constructive aspects of game theory, the exploration of instances and non-computability, and the overarching call for a \emph{more} constructive game theory has a long history. Ever since Nash's seminal contribution, Brouwer's Fixed Point Theorem has entered the foundations of game theory. Ironically\footnote{Brouwer was one of the first and strongest proponents of intuitionism -- and famous for a theorem that is not constructively valid.}, Brouwer's fixed point theorem does not admit a constructive proof. In fact, Orevkov \cite{orevkov} proved that it is false in Russian constructivism. The Weihrauch degree of Brouwer's fixed point theorem was classified in \cite{paulybrattka3}, and is the same as that of Weak K\"onig's Lemma.

The same Weihrauch degree appears behind various examples of non-computability in game theory, such as Gale-Stewart games without computable winning strategies \cite{cenzer,paulyleroux3-cie}; or behind the observation that in the infinite repeated prisoner's dilemma, there is a computable strategy that has no computable best-response (e.g.~\cite{knoblauch}, \cite{nachbar}). These examples have in common that they pertain to infinite duration games, not to finite games in normal form as we study here.

Rabin exhibited a sequential game with three rounds, with moves taken from $\mathbb{N}$, such that it is decidable who wins a given play, with the result that the second player has a winning strategy, but the second player also has no computable winning strategy \cite{rabin}. The Weihrauch degree inherent to this construction has not been properly investigated, but it already follows from Rabin's analysis that it is strictly above Weak K\"onig's Lemma.

The setting of finite games in normal form was considered constructively by Bridges in \cite{bridges}, and already showed that the minmax theorem cannot be proven in that setting as it entails the non-constructive principle $\mathrm{LLPO}$. Bridges and coauthors also delved into the construction of utility functions from preferences in a constructive framework and revealed various obstacles \cite{bridges7,bridges3,bridges4}.

Several authors (including the second author of this article) have made the case that for game theory it is more important to work constructively than for other areas of mathematics. The reason is that the solution concepts of game theory are explicitly assumed to be the result of some decision-making process by agents, which should follow the usual restrictions for computability. As discussed in \cite{paulyphd}, the requirement of falsifiability inherently constrains the level of non-constructiveness a scientific theory can exhibit. For instance, the Weihrauch degree of Weak K\"onig's Lemma is consistent with falsifiability, just not with the stronger requirement we propose for game theory. Other appeals for a more constructive approach to game theory have been presented, as demonstrated by Velupillai \cite{velupillai2,velupillai}.

\section{Computable analysis and Weihrauch reducibility}
Computability in the countable, discrete realm may be the better-known concept, through the notion of Turing computability\footnote{There also is the algebraic approach to computability as put forth by Blum, Shub, and Smale \cite{blum2}. That model does not fit the justification for why computability is required in game theory. Nash equilibria still are not computable in the BSS-model though \cite{paulybss}.} works perfectly well on most spaces of interest of cardinality up to the continuum. The field that delves into the study of computability in such settings is referred to as \emph{computable analysis}. A standard textbook is \cite{weihrauchd}. A quick introduction in a similar style is available as \cite{brattkaintro}. A concise, more general treatment is found in \cite{pauly-synthetic}.

Here, we just try to give an intuition for the notion of computability and refer to the references above for details. Turing machines inherently possess infinite tapes, allowing us to employ infinite binary sequences as their inputs and outputs. Computations then no longer halt but instead continue to produce more and more output. To get computability for interesting objects such as the reals, we code them via the infinite binary sequences. This yields the notion of a \emph{represented space}. In the case of the reals, an encoding based on the decimal or binary expansion would yield an unsatisfactory notion of computability (as multiplication by 3 would not be computable). However, an encoding via sequences of rational numbers converging with a known rate (e.g.~we could demand that $|q_n - x| < 2^{-n}$, where $x$ is coded real and $q_n$ the $n$-th approximation) works very well \cite{turingb}. Essentially, this approach renders all naturally occurring continuous functions computable. The standard representation of the real numbers is also consistent with assuming that real numbers are obtained by repeating physical measurements over and over and thus obtaining higher and higher expected accuracy \cite{paulymeasurement}.

A central notion in computable analysis is that of a multivalued function, where multiple valid outputs are permitted. Individually, these can just be defined as relations of valid input/output combinations. However, the composition of multivalued functions differs from the composition of relations. The motivation for utilising multivalued functions comes from both practical applications and the foundational model. Since inputs will often have many different names, an algorithm can easily produce names of different outputs based on different names of the same input. Nash equilibria are a good example of the motivation from applications, since a game can have multiple Nash equilibria, and we may not want to specify a particular one as the desired solution.

The framework for studying degrees of non-computability in computable analysis is Weihrauch reducibility. A multivalued function between represented spaces is Weihrauch reducible to another if there is an otherwise computable procedure invoking the second multivalued function as an oracle exactly once that solves the first. The Weihrauch degrees are the equivalence classes for Weihrauch reductions. We write $f \leqW g$, $f \leW g$ and $f \equivW g$ for $f$ being Weihrauch reducible to $g$, $f$ being strictly Weihrauch below $g$ and $f$ being Weihrauch equivalent to $g$ respectively.

The notion of Weihrauch reducibility was popularized by Brattka and Gherardi \cite{brattka3,brattka2}. A comprehensive introduction and survey are available as \cite{pauly-handbook}. Central for our results are two closure operators on the Weihrauch degrees representing allowing multiple invocations of the oracle: The degree $f^*$ represents being allowed to invoke $f$ any finite number of times \emph{in parallel}, meaning that all queries to $f$ can be computed without knowing any of the answers. The degree $f^\diamond$ represents being allowed to invoke $f$ any finite number of times (not specified in advance), where later queries can be computed from previous answers. The degree $f^\diamond$ is the least degree above $f$ which is closed under composition \cite{westrick}. Additionally, we employ the operator $\bigsqcup$, where $\bigsqcup_{n \in \mathbb{N}} f_n$ receives an $n \in \mathbb{N}$ together with an input for $f_n$, and returns a matching output.

A source of important benchmark degrees in the Weihrauch lattice is the \emph{closed choice principles}. Informally spoken, the input is information about what does {\bf not} constitute a valid output, and the output is something avoiding these obstructions. A specific closed choice principle, All-or-Unique choice, plays a central role in our investigation and will be discussed in detail below. We also make occasional reference to the \emph{finite closed choice} principle, denoted by $\C_k$. Here the ambient space is $\{0,1,\ldots,k-1\}$. An input is an enumeration of a (potentially empty) subset that excludes at least one element, valid outputs are any numbers from $\{0,1,\ldots,k-1\}$ not appearing in the enumeration. These principles form in increasing hierarchy, and already $\C_2$ (often also called $\mathrm{LLPO}$) is non-computable.

\section{Overview of our results}
Once we have established the method for representing real numbers, it is trivial to obtain a representation for finite games in strategic form and a representation for mixed strategy profiles. We can then define the multivalued function $\mathrm{Nash}$ mapping finite games in strategic form to some Nash equilibrium. We refer to the restriction of $\mathrm{Nash}$ to two-player games as $\mathrm{Nash}_2$. Our main goal is to classify the Weihrauch degree of $\mathrm{Nash}$. We do this by comparing it to a benchmark principle called All-Or-Unique Choice, $\aouc$. Essentially, $\aouc$ receives an abstract input that expresses which $x \in \uint$ are valid solutions as follows: Initially, all of $\uint$ is a valid solution. This could remain the case forever (the \emph{all}-case), or at some point, we receive the information that there is just a unique correct answer, and we are told what that one is (the \emph{unique}-case).

Another highly relevant multivalued function for us is $\BRoot$. Let $\BRoot : \mathbb{R}[X] \mto \uint$ map real polynomials to a root in $\uint$, provided there is one. $\BRoot$ is allowed to return any number $x \in \uint$ if there is no such root. Let $\BRoot_{k}$ ($\BRoot_{\leq k}$) be the restriction of $\BRoot$ of polynomial of degree (less-or-equal than) $k$. In particular, we see that $\BRoot_{\leq 1}$ is just the task of solving $bx = a$. The obstacle for this is that we do not know whether $b = 0$ and anything $x \in \uint$ is a solution, or whether $b \neq 0$ and we need to answer $\frac{a}{b}$. Following an observation by Brattka, it was shown as \cite[Proposition 8]{pauly-kihara2-mfcs} that $\aouc \equivW \BRoot_{\leq 1}$. The main result from \cite{paulyincomputabilitynashequilibria} is $\mathrm{Nash}_2 \equivW \aouc^*$. The remaining ingredient of our first main result is found in Corollary \ref{corr:nashaoucdiamond} in Subsection \ref{subsec:multivariate}: 

\begin{theorem}
\label{theo:main}
$\aouc^* \leqW \mathrm{Nash} \leqW \aouc^\diamond$
\end{theorem}

As established in \cite{pauly-kihara2-mfcs} $\aouc^* \leW \aouc^\diamond$, so at least one of the two reductions in Theorem \ref{theo:main} is strict. Furthermore, from the outcomes presented in \cite{pauly-kihara2-mfcs} it also follows that $\aouc^\diamond \equivW \aouc^* \star \aouc^*$, where $f \star g$ lets us first apply $g$, then do some computation, and then apply $f$. Consequently, any finite number of oracle calls to $\aouc$ can be rearranged to happen in two phases, where all calls within one phase are independent of each other. Drawing upon the outcome in  \cite{paulyincomputabilitynashequilibria}, we can conclude that from a multiplayer game $G$ we can compute a two-player game $G'$, take any Nash equilibrium of $G'$, and compute another two-player game $G''$ from that such that given any Nash equilibrium to $G''$ we could compute a Nash equilibrium for $G$. It seems very plausible to us that $\aouc^* \equivW \mathrm{Nash}$ should hold, but constructing a proof for this equivalence has posed a challenge. Thankfully, the outcomes discussed in Section \ref{sec:consequences}, stemming from Theorem \ref{theo:main}, do not hinge on the resolution of these specific details.

The lower bound in Theorem \ref{theo:main} clearly already follows from $\aouc^* \equivW \mathrm{Nash}_2$. For the upper bound, we bring three ingredients together. The first result shows that a preprocessing step that identifies a potential support for a Nash equilibrium can be absorbed into the Weihrauch degree $\aouc^*$. Once we know the support we are going to use, we have a solvable system of polynomial inequalities whose solutions are all Nash equilibria. 

We are thus led to investigate the Weihrauch degrees pertaining to solving systems of polynomial (in)equalities. We completely classify the degree of finding (individual) roots within given bounds of finitely many (univariate) polynomials:

\begin{theorem}[Proven as consequence of Corollary \ref{corr:broot} below]
\label{theo:polys}
$\aouc^* \equivW \BRoot^*$
\end{theorem}

Consequently, disregarding the precise count of required oracle calls, the task of finding polynomial roots does not pose a greater challenge than solving equations in the form $bx = a$. The proof of Theorem \ref{theo:polys} in turn makes use of a result from \cite{paulyleroux} on finding zeros of real functions with finitely many local minima.

To prove Theorem \ref{theo:main} we need more, namely we need to solve systems of (in)equalities for multivariate polynomials. This aspect is addressed in the forthcoming Corollary \ref{corr:cadmain}. Through an examination of cylindrical algebraic decomposition, augmented with minor adaptations, we demonstrate that access to $\aouc^*$ lets us compute finitely many candidate solutions including a valid one. A comprehensive introduction to cylindrical algebraic decomposition (CAD) is available in \cite{Jirstrand}. CAD involves the partitioning of $\mathbb{R}^n$ into semi-algebraic cylindrical cells.

%This algorithm has two main phases, projection and extension. The projection operator takes polynomials in $k$ variables and produces another set in $k-1$ variables repeatedly until it identifies points on the real line. This part is effective. We can look at the intervals between 0 and 1 and use CAD to break them down. Using finite closed choice as a preprocessing step we can choose the broken down intervals that work. We can then use $\aouc^\diamond$ to locate the specific interval of the root of the polynomial if one exists.

\section{Roots of polynomials}
In this section, we consider the computability aspects of finding the roots of polynomials. Our exploration of polynomial root finding unfolds across three distinct scenarios. In Subsection \ref{subsec:monic} we look into monic univariate polynomials. In Subsection \ref{subsec:univariate} we drop the restriction to monic, and in Subsection \ref{subsec:multivariate} we handle multivariate polynomials.

In order to conceptualise the task of polynomial root finding, it becomes essential to establish a clear methodology for representing polynomials. A polynomial is represented by providing an upper bound to its degree, plus a tuple of real numbers constituting all relevant coefficients. For example, we could be given the same polynomial as either $0x^2 + 3x - 5$ or $3x - 5$. If we were demanding to know the exact degree, we could no longer compute the  addition of polynomials and the multiplication of polynomials with real numbers, which would be clearly unsatisfactory. This matter is discussed in detail in \cite[Section 3]{pauly-steinberg}. We denote the represented space of real univariate polynomials as $\mathbb{R}[X]$ and the space of real multivariate polynomials as $\mathbb{R}[X^*]$. For the latter, we assume that each polynomial comes with the exact information of what finite set of variables it refers to. Both $\mathbb{R}[X]$ and $\mathbb{R}[X^*]$ are coPolish spaces, see \cite{schroder6,paulydebrecht4,callard}.

\subsection{Monic univariate polynomials}
\label{subsec:monic}
It is known since the dawn of computability theory that the fundamental theorem of algebra is constructive, more precisely, that given a monic polynomial with real (or complex) coefficients, we can compute the unordered tuple of its complex roots, each repeated according to its multiplicity. The latter formulation was established by Specker \cite{specker-fta}.

Of course, we cannot decide which of the complex roots are real. The task of selecting a real number from a $k$-tuple of complex numbers containing at least one real number is Weihrauch equivalent to $\C_k$. A similar task equivalent to $\C_k$ is to identify a $0$-entry in a $k$-tuple of real numbers containing at least one $0$. Given real numbers $\varepsilon_0, \varepsilon_1,\ldots,\varepsilon_{k-1}$ we can construct the monic polynomial $\Pi_{i < k}((x - \frac{i}{k})^2 + |\varepsilon_i|)$, which will have a root at $\frac{i}{k}$ iff $\varepsilon_i = 0$. This shows that finding real roots of monic polynomials is Weihrauch equivalent to $\C_2^* = \bigsqcup_{k \in \mathbb{N}} \C_k$.  

\subsection{Univariate polynomials}
\label{subsec:univariate}
In general, we do not know the degree of a polynomial and thus cannot restrict ourselves to the monic case for root-finding. It has been shown by Le Roux and Pauly \cite{paulyleroux} that knowing a finite upper bound $k$ on the number of local extrema lets us find a root of a continuous function (if it has one in a bounded interval) using $\C_{3^k}$. Essentially, this observation suggests that for polynomial root finding in $\uint$, knowing the precise degree versus knowing an upper bound makes only a quantitative, but not a qualitative difference -- if we exclude the zero polynomial!

While it may seem counterintuitive that it should be the zero polynomial that makes root finding more difficult, we will see that this is the case, and explore how much.

\begin{definition}
Let $\BRoot : \mathbb{R}[X] \mto \uint$ map real polynomials to a root in $\uint$, provided there is one, and to arbitrary $x \in \uint$ otherwise. Let $\BRoot_{k}$ ($\BRoot_{\leq k}$) be the restriction of $\BRoot$ to polynomials of degree (less than -or-equal to) $k$.
\end{definition}

We have defined $\BRoot$ to be a total map. If the input is a polynomial without a root in the unit interval, it will return an arbitrary element of the unit interval instead. However, the restriction of $\BRoot$ which is defined only on polynomials with a root in the unit interval is in fact equivalent to $\BRoot$ itself. It's worth noting that this is a special case of Lemma~\ref{lemma:makezero} which we will prove later.

\begin{proposition}\label{prop4}
$\BRoot_{\leq 2k+1} \leqW \aouc \times \C_{3^k}$
\begin{proof}
A polynomial $p$ of degree at most $2k+1$ is either the 0 polynomial, or has at most $k$ local minima. If it is not $0$, we will recognize this, and moreover, can find some $a \leq 0$, $b \geq 1$ with $p(a) \neq 0 \neq p(b)$. In this case, \cite[Theorem 4.1]{paulyleroux} applies and lets us compute a $3^k$-tuple of real numbers amongst which all zeros of $p$ between $a$ and $b$ will occur.

We recall that $\aouc \equivW \textrm{AoUC}_{[a,b]^n}$ \cite[Corollary 12]{pauly-kihara2-mfcs} for all $n > 1$. We let the input to $\textrm{AoUC}_{[a,b]^{3^k}}$ be $[a,b]^{3^k}$ as long as $p = 0$ is consistent, and if we learn that $p \neq 0$, we collapse the interval to the $3^k$-tuple we obtain from \cite[Theorem 4.1]{paulyleroux}. The input to $\C_{3^k}$ is initially $\{0,\ldots,3^k-1\}$. We only remove elements after we have confirmed $p \neq 0$, and then we remove $j$ if the $j$-th candidate obtained from \cite[Theorem 4.1]{paulyleroux} is not a root of $p$, or falls outside of $[0,1]$.

To obtain the answer to $\BRoot_{\leq 2k+1}$, we just use the answer from $\C_{3^k}$ to indicate the answer of which component of the tuple provided by $\textrm{AoUC}_{[a,b]^{3^k}}$ to use as the final output.
\end{proof}
\end{proposition}

The following is also a direct consequence of \cite[Theorem 12]{pauly-kihara5-arxiv}, but its proof is much more elementary; and a direct consequence of \cite[Proposition 17.4]{hoelzl}, but again with a simpler proof:
\begin{proposition}
\label{prop:noparallel}
$\aouc \times \C_2 \nleqW \BRoot$
\begin{proof}
Assume that the reduction would hold. Let $q$ be a name for an input to $\aouc$ which never removes any solutions (and never commits to not removing any solution). Assume further that there is some name $r$ for an input to $\C_2$ such that $(q,r)$ gets mapped to a non-zero polynomial $p$ by the inner reduction witness. Since $p$ is ensured to be non-zero already by a finite prefix of its names, and since the inner reduction witness is continuous, it holds that there is an $M \in \mathbb{N}$ such that any $(q',r)$ with $d(q,q') < 2^{-M}$ gets mapped to a non-zero polynomial. Restricting $\aouc$ to names from $\{q' \mid d(q,q') < 2^{-M}\}$ does not change its Weihrauch degree (as $q$ provides no information at all). By \cite[Corollary 4.3]{paulyleroux}, if we exclude the $0$ polynomial from the domain of $\BRoot$, then $\C_2^*$ suffices to find a root. We can thus conclude $\aouc \leqW \C^*_2$, but this contradicts \cite[Theorem 22]{paulyincomputabilitynashequilibria}.

Thus, it would need to hold that each pair $(q,r)$ gets mapped to the $0$ polynomial. But since answering constant $0$ is a computable solution to $\BRoot(0)$, this, in turn, would imply that $\C_2$ is computable, which is absurd. We have thus arrived at the desired contradiction.
\end{proof}
\end{proposition}

While our preceding proposition shows the limitations of $\BRoot$ for solving multiple non-computable tasks in parallel, the following result from the literature reveals that the slightly more complex nature of $\aouc$ is central. Note that $\C_n \times \C_m \leqW \C_{n\cdot m}$, hence $\C_k$ has an inherently parallel nature for $k \geq 4$.

\begin{proposition}[{\cite[Proposition 4.6]{paulyleroux}}]
\label{prop:leroux46}
$\C_k \leqW \BRoot_{2k}$
\end{proposition}

\begin{corollary}
\label{corr:broot}
$\aouc \leW \BRoot \leW \aouc^*$
\begin{proof}
The first reduction follows from $\aouc \equivW \BRoot_{\leq 1}$ (\cite[Proposition 8]{pauly-kihara2-mfcs}). That it is strict comes from \cite[Proposition 4.6]{paulyleroux} showing that otherwise, we would have $\C_3 \leqW \aouc \leqW \lpo$, contradicting a core result from \cite{weihrauchc}.

The second reduction follows from Proposition~\ref{prop4}, together with $\C_{k+1} \leqW \C_2^k$ and $\C_2 \leqW \aouc$ (both from \cite{paulyincomputabilitynashequilibria}). Its strictness is a consequence of Proposition \ref{prop:noparallel}.
\end{proof}
\end{corollary}

\subsection{Multivariate Polynomials and Cylindrical Algebraic Decomposition}
\label{subsec:multivariate}
Our focus now shifts towards multivariate polynomials. Within this domain, two intriguing problems arise: the pursuit of a root for an individual multivariate polynomial and the search for a solution to a system of polynomial inequalities (restricting ourselves to the non-strict case for now). In both scenarios, our emphasis rests on the bounded case, given its relevance to Nash equilibria. Additionally, the unbounded case promptly leads us to $\lpo^*$.

\begin{definition}
Let $\BMRoot : \mathbb{R}[X^*] \mto \uint^*$ map real polynomials to a root in $\uint^*$, provided there is one; and to an arbitrary $x \in \uint^*$ otherwise.
\end{definition}

Our definition of $\BMRoot$ extends to all polynomials, and when confronted with a lack of roots within the unit hypercube, it provides an arbitrary element from the unit hypercube as an output. The latter scenario bears lesser significance, Lemma \ref{lemma:makezero} detailed below shows that restricting $\BMRoot$ to polynomials having roots in the unit hypercube does not change its Weihrauch degree.

\begin{definition}
Let $\mathrm{BPIneq} : (\mathbb{R}[X^*])^* \mto \uint^*$ map a sequence of polynomials $P_1,\ldots,P_k$ to a point $x \in \uint^*$ such that $\forall i \leq k, \ P_i(x) \geq 0$ if such a point exists, and to any $x \in \uint^*$ otherwise.
\end{definition}

\begin{lemma}
\label{lemma:makezero}
There is a computable multivalued map $\mathrm{MakeZero} : \mathbb{R}[X^*] \mto \mathbb{R}[X^*]$ such that whenever $g \in \mathrm{MakeZero}(f)$, then $g$ has a zero in the unit hypercube; and if $f$ already had a zero in the unit hypercube, then $f = g$.
\end{lemma}
\begin{proof}
Given an $n$-variate polynomial $f$, we can compute $c := \min \{|f(x)| \mid x \in \uint^n\}$ since the unit hypercube is computably compact and computably overt. Let $\mathbb{T}$ be Plotkin's $T$, i.e.~the space with the truth values $0$, $1$ and undefined $(\bot)$. We can compute $\operatorname{sign}(f(0^n)) \in \mathbb{T}$ (where $\operatorname{sign}(1) = 1$, $\operatorname{-1} = 0$ and we consider the sign of $0$ to be undefined).

The operation $\operatorname{Merge} : \subseteq \mathbb{T} \times \mathbf{X} \times \mathbf{X} \to \mathbf{X}$ is defined on all inputs except $(\bot,x,y)$ for $x \neq y$, and satisfies $\operatorname{Merge}(0,x,y) = x$, $\operatorname{Merge}(1,x,y) = y$ and $\operatorname{Merge}(\bot,x,y) = x = y$. It is easy to see that $\operatorname{Merge}$ is computable for $\mathbf{X} = \mathbb{R}$. We use it to compute $\overline{c} = \operatorname{Merge}(\operatorname{sign}(f(0^n)), c, -c)$ and find that $f + \overline{c}$ meets the requirements to be an output to $\mathrm{MakeZero}(f)$. Essentially, we just shift $f$ vertically by the minimal amount required to make it have a zero in the unit hypercube.
\end{proof}

\begin{proposition}
\label{prop:bmrootidempotent}
$\BMRoot \equivW \BMRoot^*$.
\end{proposition}

\begin{proof}
It suffices to show that $\BMRoot^2: \mathbb{R}[X^*] \times \mathbb{R}[X^*] \mto \uint^* \times \uint^*$ is reducible to the map $\BMRoot: \mathbb{R}[X^*]  \mto \uint^*$. We are given two multivariate polynomials $P$ and $Q$. We rename variables, in order to ensure that each polynomial uses different variables. By Lemma \ref{lemma:makezero} we can assume w.l.o.g.~that $P$ and $Q$ each have a root in the unit hypercube.

We then apply $\BMRoot$ to $P(\overline{x})^2 + Q(\overline{y})^2$ and obtain a root $(\overline{x}_0,\overline{y}_0)$ of the latter polynomial. Now, $\overline{x}_0$ is a root of $P$ and $\overline{y}_0$ is a root of $Q$.

\end{proof}

%\begin{proposition}
%$\mathrm{BPIneq} \equivW \mathrm{BPIneq}^*$
%\end{proposition}
%\begin{proof}
%Next, we need to show $\mathrm{BPIneq}^2: %(\mathbb{R}[X^*])^* \times (\mathbb{R}[X^*])^* \mto \uint^* \times \uint^*$ is reducible to $\mathrm{BPIneq} : (\mathbb{R}[X^*])^* \mto \uint^*$. We can combine two lists of polynomials simply, resulting in a longer list which would have the same properties as the originals. 

%\end{proof}

We will conclude that $\aouc^\diamond$ is an upper bound for the Weihrauch degree of $\mathrm{BPineq}$ (and thus $\BMRoot$) as a corollary of the main theorem of this subsection, Theorem \ref{theo:cad}. The way we obtain Theorem \ref{theo:cad} is via an analysis of how constructive CAD is if we do not take the equality test on the reals for granted.

To tackle the task of finding the number of roots in multivariate polynomials, even in the absence of prior knowledge about the degrees of these polynomials, we implement the CAD algorithm. This approach enables us to systematically deconstruct multivariate polynomials into lower-variate forms. This can be repeated until we have a set of univariate polynomials. Subsequently, we can find the roots of these univariate polynomials in the same way as Section \ref{subsec:univariate}. This establishes a foundation for determining the set of solutions within a system of polynomial equations and inequalities. The subsequent step involves `lifting' these univariate polynomials to the original multivariate setting to pinpoint their roots within the context of the overarching problem. 

The CAD algorithm has three phases; projection, base, and extension. This algorithm is driven by an input, represented as a set $\F$ of $n$-variate polynomials. In the projection phase, a sequence of $n-1$ steps is executed, each resulting in the creation of new polynomial sets. The zero set of the resulting polynomials consists of the projection of the significant points. The base phase isolates the real roots of the univariate polynomials from the outputs of the projection phase. Each root and one point in the intervals between roots are then used as sample points in the decomposition of $\R^1$. The extension phase constructs sample points for all regions of the CAD of $\R^n$. This phase also consists of $n-1$ steps which takes the sample points from the base phase and `lifts' them into $\R^2$ for each region in the stack. This is then repeated until we have sample points to all regions of the CAD of $\R^n$.

In more detail, the procedure involves the upward `lifting' of univariate polynomials from $\mathbb{R}^{i-1}$ to $\mathbb{R}^i$. This elevation is achieved by evaluating the polynomials in $\proj^{n-i+1}(\mathcal{F})$ over a sample point $\alpha$. This results in a set of univariate polynomials in $x_i$ corresponding to the values of $\proj^{n-i+1}(\mathcal{F})$ on the `vertical' line $x_{i-1} = \alpha$. These univariate polynomials are treated the same in each `lifting' phase until they reach $\mathbb{R}^n$.

%The CAD algorithm uses the concept of regions, cylinders, decomposition's and stacks. In essence, this allows us to split any polynomial into regions which can be decomposed into sections and sectors. For example, the polynomial $y=x^2$ would have one section for the root at (0,0) and two sectors $x>0,~x<0$ in a decomposition of $\mathbb{R}^1$. This intuitively shows each region is a connected subset of $\mathbb{R}^n$. When we have three or more variables, these regions are stacked within a cylinder in such a way that there are no overlaps between sections or sectors.

\begin{definition}\label{def:sec}
\begin{itemize}
\item A region $\mathcal{R}$ is a connected subset of $\mathbb{R}^n$.
\item The set $Z(\mathcal{R}) = \mathcal{R} \times \mathbb{R} = \{(\alpha,x) \mid \alpha \in \mathcal{R}, x \in \mathbb{R}\}$ is called a cylinder over $\mathcal{R}$.
\item Let $f, f_1, f_2$ be continuous, real-valued functions on $\mathcal{R}$. An $f$-section of $Z(\mathcal{R})$ is the set $\{(\alpha,f(\alpha)) \mid \alpha \in \mathcal{R}\}$ and an $(f_1,f_2)$-sector of $Z(\mathcal{R})$ is the set $\{(\alpha,\beta) \mid \alpha \in \mathcal{R}, f_1(\alpha) < \beta < f_2(\alpha)\}$.
\end{itemize}
\end{definition}

Within the context of CAD, the regions from $\mathcal{R}$ that make an appearance signify the locations where the $n+1$-variate polynomial possesses a root (with the first $n$-variables ranging over $\mathcal{R}$). Specifically, the first $n$ variables range over $\mathcal{R}$ in this process. Concurrently, the sectors encapsulate the intermediary segments where the polynomial remains consistently positive or negative. This contributes to a decomposition, which entails the fragmentation of a given region into smaller, distinct components. Therefore, a decomposition of $\mathbb{R}$ is a set of regions, defined above as sectors and sections. 

\begin{definition}\label{def:decomp}
Let $\mathcal{R} \subseteq \mathbb{R}^n$. A decomposition of $\mathcal{R}$ is a finite collection of disjoint regions (components) whose union is $\mathcal{R}$: 
$\mathcal{R} = \bigcup^k_{i=1} \mathcal{R}_i, ~ \mathcal{R}_i \cap \mathcal{R}_j = \emptyset$ whenever $i \neq j$.

A stack over $\mathcal{R}$ is a decomposition of $\mathcal{R} \times \mathbb{R}$ comprising a combination of $f_i$-sections and $(f_i,f_{i+1})$-sectors, where $f_0 < \dots < f_{k+1}$ for all $x \in \mathcal{R}$ and $f_0 = -\infty, f_{k+1} = +\infty$. 

The stack decomposition of $\mathbb{R}^0 = \{0\}$ is $\{\{0\}\}$. A stack decomposition of $\mathbb{R}^{n+1}$ is a decomposition of the form $\bigcup_{\mathcal{R} \in \mathfrak{D}} S_\mathcal{R}$, where each $S_\mathcal{R}$ is a stack over $\mathcal{R}$, and $\mathfrak{D}$ is a stack decomposition of $\mathbb{R}^n$.
\end{definition}

The initial phase revolves around projecting polynomials from $n$ variables to a set in $n-1$ variables. Within this process, a real polynomial $f_i \in \mathbb{R}[x_1,\dots,x_{n-1}][x_n]$ in $n$-variables can be deconstructed into the coefficients for every power of $x$: $ f_i(x_1,\dots,x_{n-1},x_n) = f^{d_i}_i(x_1,\dots,x_{n-1})x^{d_i}_n +\dots + f^0_i(x_1,\dots,x_{n-1})$, where $d_i$ signifies the degree of the polynomial.

\begin{definition} \label{def:red}
The \emph{reductum}, $\hat{f}^{k_i}_i$ of a polynomial is
$$ \hat{f}^{k_i}_i(x_1,\dots,x_{n-1},x_n)=f^{k_i}_i(x_1,\dots,x_{n-1})x^{k_i}_n + \dots + f^0_i(x_1,\dots,x_{n-1})$$
where $0 \leq k_i \leq d_i$.
\end{definition}

\begin{definition} \label{def:psc}
Let $f,g \in \mathbb{R}[x]$ and $\textrm{deg}(f) = m, \textrm{deg}(g) = n, m \geq n$. The $k^{\textrm{th}}$ \emph{principal subresultant coefficient} of f and g is $$ \psc_k(f,g)=\textrm{det}(\mathrm{M}_k),~0\leq k\leq n$$
where $\mathrm{M}_0$ is the Sylvester matrix of f and g, and then $\mathrm{M}_k$ is obtained by deleting certain rows and columns from $\mathrm{M}_0$.
\end{definition}

The concrete definitions of the $\mathrm{M}_k$ matrices are not of direct relevance to our current discussion, but they are discussed in depth in \cite[Example 4.4]{Jirstrand}.

\begin{lemma}\label{lem:comp}
\begin{enumerate}
\item{The reductum of polynomials is computable.}
\item{The derivative of a polynomial is computable.}
\item{Given two polynomials $p,q$, we can compute a finite tuple of polynomials such that every well-defined $\psc_k(p,q)$ appears within the tuple.}
\end{enumerate}
\end{lemma}

\begin{proof}
\begin{enumerate}
\item{Calculating the reductum requires the rearrangement of the polynomial $\hat{f}^{k_i}_i$, in order to group all of the coefficients for every power of $x_n$, $x^{k_i}_n,...,x^0_n$. This is trivially computable, as no tests on the coefficients need to be performed.}
\item{Calculating the derivative of a polynomial merely requires the multiplication of coefficients with natural numbers.}
\item{As we do not have access to the exact degrees of the polynomials, but merely to some upper bound, we do not even know which of the psc's are well-defined. Let $n$ be the upper bound of $\textrm{deg}(p)$ and $m$ be the upper bound of $\textrm{deg}(q)$. We have a finite potential number of pairs of polynomials, $(n+1) \times (m+1)$ combinations. These can be used to calculate the $\psc$, under the assumption that the current degree of the pair of polynomials are the correct degrees.}
\end{enumerate}
\end{proof}

We point out that the use of an ``overapproximation'' in Lemma~\ref{lem:comp} (3) is unavoidable. If we knew how many principal subresultant coefficients there are for $f$ and $g$, we would know the degree of $g$.

%Given a set of polynomials $\F \subset \mathbb{R}[x_1,\dots,x_{n-1}][x_n]$ we can compute another set of ($n-1$)-variate polynomials $\proj(\F) \subset \mathbb{R}[x_1,\dots,x_{n-1}]$, which characterises the maximal connected $\F$-delineable sets of $\mathbb{R}^{n-1}$. The projection can be summarised by 
%$$\F_0 = \F \subset \mathbb{R}[x_1,\dots, x_{n-1}, x_n], \F_1 \subset \mathbb{R}[x_1,\dots,x_{n-1}], \dots, \F_{n-1} \subset \mathbb{R}[x_1]$$

\begin{definition} \label{def:proj}
Let $\F=\{f_1,f_2,\dots,f_r\}$ be a finite set of $n$-variate polynomials. Its projection $\proj(\F) =\proj_1(\F) \cup \proj_2(\F) \cup \proj_3(\F)$ is a finite set of $n-1$-variate polynomials, where 
\begin{align*}
\proj_1 & =\{f^k_i(x_1,\dots,x_{n-1}) \mid 1 \leq i \leq r, 0 \leq k \leq d_i\} \\
\proj_2 & =\{\psc^{x_n}_l (\hat{f}^k_i(x_1,\dots,x_{n-1}), 
\textrm{D}_{x_n}(\hat{f}^k_i(x_1,\dots,x_{n-1}))) \mid 1\leq i \leq r, 0 \leq l <k\leq d_i -1\} \\
\proj_3 &= \{\psc^{x_n}_m(\hat{f}^{k_i}_i(x_1,\dots,x_{n-1}),\hat{f}^{k_j}_j(x_1,\dots,x_{n-1})) \mid \\ & \quad 1\leq i<j\leq r, 0\leq m\leq k_i\leq d_i, 0\leq m\leq k_j \leq d_j\}.
\end{align*} 
Here $\psc^{x_n}_i$ denotes the $i^{th}$ principle resultant coefficient w.r.t $x_n$ and $\textrm{D}_{x_n}$ is the formal derivative operator with respect to $x_n$.
\end{definition}

\begin{corollary}\label{corr:proj}
From a finite tuple of polynomials $p_1,p_2,\ldots,p_k$, we can compute a finite tuple of polynomials $q_1,q_2,\ldots,q_\ell$ such that $\proj(\{p_i \mid i \leq k\}) \subseteq \{q_j \mid j \leq \ell\}$.
\end{corollary}

\begin{proof}
By Lemma \ref{lem:comp}.
\end{proof}

\begin{definition} \label{def:del}
Let $\F = \{f_1,f_2,\dots, f_r\} \subset \mathbb{R}[x_1,\dots,x_{n-1}][x_n]$ be a set of multivariate real polynomials and $\mathcal{R} \subseteq \mathbb{R}^{n-1}$ be a region. We say that $\F$ is delineable on $\mathcal{R}$ if it satisfies the following invariance properties:
\begin{enumerate}
\item{For every $1\leq i \leq r$, the total number of complex roots of $f_{i}(y)$ remains invariant as $y$ varies over $\mathcal{R}$.}
\item{For every $1 \leq i \leq r$, the number of distinct complex roots of $f_{i}(y)$ remains invariant as $y$ varies over $\mathcal{R}$.}
\item{For every $1 \leq i \leq j \leq r$, the total number of common complex roots of $f_{i}(y)$ and $f_{j}(y)$ remains invariant as $y$ varies over $\mathcal{R}$.}
\end{enumerate}
\end{definition}

\begin{definition}
For a list of $n$-variate polynomials $f_1,f_2,\dots,f_r$ and $x \in \mathbb{R}^n$ let $\operatorname{sign}(f_1,f_2,\dots,f_r,x) \in \{0,+,-\}^r$ be defined by $\operatorname{sign}(f_1,f_2,\dots,f_r,x)(j) = \mathalpha{+}$ if $f_j(x)>0$, $\operatorname{sign}(f_1,f_2,\dots,f_r,x)(j) = \mathalpha{-}$  if $f_j(x)<0$ and $\operatorname{sign}(f_1,f_2,\dots,f_r,x)(j) = 0$  if $f_j(x)=0$.
\end{definition}

If a set of polynomials is delineable over a region $\mathcal{R}$, then the sign vector remains invariant over $\mathcal{R}$, \cite[Lemma 1]{Jovanovic}.

\begin{lemma} \label{lem:del}
Let $\mathcal{F}  = \{f_1,\ldots,f_r\} \subset \mathbb{R}[x_1,\dots,x_{n-1}][x_n]$ be a set of polynomials, and let $proj(\mathcal{F})  = \{q_1,\ldots,q_r\} \subset \mathbb{R}[x_1,\dots,x_{n-1}]$ be the set of its projections. For any $b \in \{0,+,-\}^r$ we find that $\mathcal{F} $ is delineable on $R_b := \{x \mid \operatorname{sign}(q_1,\ldots,q_r,x) = b\}$.
\end{lemma}

\begin{proof}
Following \cite[Proof to Theorem 4.1]{Jirstrand} we show that the three properties (total number of complex roots, number of distinct complex roots, and number common complex roots) required to be invariant in Definition \ref{def:del} can be expressed by referring to the signs of polynomials belonging to the projections.

\begin{enumerate}
    \item{ Total number of complex roots of $f_i(y)$ remains invariant over $\mathcal{R}$. This is expressed by
    \begin{equation*}
        (\exists 0 \leq k_i \leq d_i)
        \big[(\forall k > k_i)[f_i^k(x_1,\dots,x_{n-1})=0]\wedge f_i^{k_i}(x_1,\dots,x_{n-1}) \neq 0 \big]
    \end{equation*}
    holding for all $y \in \mathcal{R}$.}
    \item{``The number of distinct complex roots of $f_i(y)$ remains invariant over $\mathcal{R}$.'' is equivalent to:
    \begin{align*}
        (\exists 0 < k_i \leq d_i)&(\exists 0 \leq l_i \leq k_i -1) \\
        &\big[ (\forall k > k_i)[f_i^k(x_1,\dots,x_{n-1})=0] \wedge f_i^{k_i}(x_1,\dots,x_{n-1})\neq 0 ~\wedge \\
        &(\forall l < l_i)[\psc_l^{x_n}(\hat{f}_i^{k_i}(x_1,\dots,x_{n-1}),\textrm{D}_{x_n}(\hat{f}_l^{k_i}(x_1,\dots,x_n)))=0]~\wedge \\
        &\psc_{l_i}^{x_n}(\hat{f}_i^{k_i}(x_1,\dots,x_n),\textrm{D}_{x_n}(\hat{f}_i^{k_i}(x_1,\dots,x_n)))\neq 0 \big]
    \end{align*}
    holding for all $y \in \mathcal{R}$.}
    \item{``The total number of common complex roots of $f_i(y)$ and $f_j(y)$ remains invariant over $\mathcal{R}$.'' is equivalent to:
    \begin{align*}
        (\exists 0 < k_i \leq d_i)&(\exists 0 < k_j \leq d_j)(\exists 0 \leq m_{i,j} \leq \textrm{min}(d_i,d_j)) \\
        &\big[ (\forall k > k_i)[f_i^k(x_1,\dots,x_{n-1})=0]\wedge f_i^{k_i}(x_1,\dots,x_{n-1})\neq 0~ \wedge \\
        & (\forall k>k_j)[f_j^k(x_1,\dots,x_{n-1})=0]\wedge f_j^{k_j}(x_1,\dots,x_{n-1})\neq 0~ \wedge \\
        & (\forall m < m_{i,j})[\psc_m^{x_n}(\hat{f}_i^{k_i}(x_1,\dots,x_n),\hat{f}_j^{k_j}(x_1,\dots,x_n))=0] ~\wedge \\
        &[\psc_{m_{i,j}}^{x_n}(\hat{f}_i^{k_i}(x_1,\dots,x_n),\hat{f}_j^{k_j}(x_1,\dots,x_n))\neq 0]\big]
    \end{align*}
    holding for all $y \in \mathcal{R}$.}
\end{enumerate}

\end{proof}

%Using Definition \ref{def:del} the projections have invariant properties. For $\proj_1(\F)$ this implies that the degree w.r.t $x_n$ of the polynomials is constant over $C$ if $\proj(\F)$ is invariant over $C \subset \R^{n-1}$ and hence the number of roots of each polynomial is constant over $C$. For $\proj_2(\F)$ it implies that the gcd of each polynomial and its derivative has a constant degree (\cite[Lemma 4.4]{Jirstrand}, \cite[Corollary 7.7.9]{mishra_1993}). This in combination with $\proj_1(\F)$ shows that the number of distinct zeros of each polynomial in $\F$ is constant \cite[Lemma 4.2]{Jirstrand}. 

%When $\proj_3(\F)$ is considered along with $\proj_1(\F)$ the invariant property implies that the number of common zeros of each pair of polynomials in $\F$ is constant \cite[Lemma 4.4]{Jirstrand}.  The set $\proj(\F)$ characterises the sets over which are a constant number of real zeros of the polynomial in $\F$ (\cite[Lemma 4.5]{Jirstrand}, \cite[Lemma 8.6.3]{mishra_1993}).

\begin{lemma}\label{lem:dec}
For each finite set $\F$ of $n$-variate polynomials there exists a sign invariant stack decomposition of $\mathbb{R}^n$.
\end{lemma}
\begin{proof}
%The baseline, $\mathbb{R}^1$, is connected and can easily be split into sections and sectors (Definition \ref{def:sec}) creating a sign invariant decomposition. We consider all real root functions and consider the stacks defined over them. These stacks are also a decomposition (Definition \ref{def:decomp}) into sections and sectors, where the sections are defined by the roots of the polynomial and the sectors are the intervals in between them. The sectors and sections within this stack decomposition can then be treated the same way, and raised to the next dimension. This repeats until it reaches $\mathbb{R}^n$.

The case $n = 1$ is immediate; we simply partition $\mathbb{R}$ into the roots of the polynomials and the open intervals determined by them. 

Otherwise, in the base phase of the CAD algorithm we repeatedly apply the projection operator $n-1$-times. We then obtain a decomposition of $\mathbb{R}^1$ which is sign-invariant for $\proj^{n-1}(\F)$.

We can extend a stack decomposition $\mathcal{D}_{i-1}$ of $\mathbb{R}^{i-1}$ which is sign invariant for $\proj^{n-i+1}(\F)$ to a stack decomposition $\mathcal{D}_i$ of $\mathbb{R}^i$ which is sign invariant for $\proj^{n-i}(\F)$:
By Lemma \ref{lem:del} $\proj^{n-i+1}(\F)$ is delineable over each region of $\mathcal{D}_{i-1}$ and hence the real roots of $\proj^{n-i+1}(\F)$ vary continuously over each region of $\mathcal{D}_{i-1}$, while maintaining their order (cf.~\cite[Corollary 8.6.5]{mishra_1993}).
\end{proof}

We can now define what we seek to compute:

\begin{definition}
A representative sample for a finite set $\F$ of $n$-variate polynomials is a finite set of points $X$, such that for every non-empty region $\mathcal{R}$ of the sign invariant stack decomposition provided by Lemma \ref{lem:dec} there exists some $x \in \mathcal{R} \cap X$. 
\end{definition}

\begin{lemma}\label{lem:rep}
Let $X$ be a representative sample for $\proj(f_1,f_2,\dots,f_r)$. Then there is a representative sample $X'$ for $f_1,f_2,\dots,f_r$ such that $X = \{(x_1, \dots, x_{n-1}) \mid \exists x_n . (x_1, \dots, x_{n-1}, x_n) \in X' \}$. Moreover, $X'$ can be obtained as follows: For each $\overline{x} \in X$, let $S_{\overline{x}}$ be a representative sample for the univariate polynomials $f_1(\overline{x}),\ldots,f_r(\overline{x})$, and then let $X' = \{(\overline{x},x_n) \mid \overline{x} \in X \wedge x_n \in S_{\overline{x}}\}$.
\end{lemma}

\begin{proof}
In the representative sample $X=\{a_1,\dots,a_q\}$, each $a_i$ is a set of points for each region of the projections, therefore a root or a point within the interval between two non-trivial roots. We select a set of test points $b = (b_1,\dots,b_p)$ for each region of the original polynomials.

We will construct the sample points of the regions of $\mathcal{D}_i$ which belong to the stack over the region $\mathcal{C} \subset \mathcal{D}_{i-1}$. The polynomials in $\proj^{n-1}(\F)$ can be evaluated over the sample point, resulting in a set of univariate polynomials in $x_i$. These univariate polynomials can now be treated the same as the projections, where we can isolate the roots and a representative sample.  

In order to extend this to $\mathbb{R}^n$, we can substitute $a_i$ into our original polynomial in order to achieve a set of $r$ univariate polynomials in $x_n$, which we will denote $F_a = \{f_1(a_i),\dots,f_n(a_i)\}$. We can also substitute in $\proj(b)$ for a second set of $r$ univariate polynomials, which we will denote $F_b = \{f_1(\proj(b)),\dots,f_n(\proj(b))\}$. This will allow us to compare two polynomials, one from each set $F_a, F_b$, in order to find a suitable point of $x_n$ which will result in the same $\operatorname{sign}$ vector. As our test points $b$ are already in $\mathbb{R}^n$, we already know $\operatorname{sign}(f_1,\dots,f_r,b)$ and hence know what sign we need the univariate polynomials, $F_a$ to be for the $\operatorname{sign}$ vectors to be equal. 

By Lemma \ref{lem:del}, the choice of the points $a_i$ makes sure the univariate polynomials $F_a$ have the same total number of complex roots, and the same number of distinct complex roots as our original polynomials \footnote{\cite{Jirstrand}'s proof to Theorem 4.1 excludes the zero polynomial, however, the argument still works.}. If the total number of complex roots is odd, the polynomial has at least one real root with an odd multiplicity. If the total number of complex roots is even, then there could be no real roots or real roots with even multiplicity. 

For an $F_b$ with all three states (positive, negative, and zero sign vector), we would need to confirm our $F_a$ can also take all three states. This can be achieved by checking the multiplicity of the roots. The polynomial has a root with an odd multiplicity iff it crosses the axis. We know what sign the original polynomial will give from $\operatorname{sign}(f_1,\dots,f_r,b)$ allowing us to find the condition needed on the point $x_n$ to give $F_a$ in order for $\operatorname{sign}(f_1,\dots,f_r,(a_i,x_n)) = \operatorname{sign}(f_1,\dots,f_r,b)$.

If it is the zero polynomial, we can select any point as our $x_n$ and the sign vectors would be equal.

If the univariate polynomials have an even multiplicity they would either have a sign vector $\{0,+\}$ or $\{0,-\}$. We would confirm what sign this need to be using $\operatorname{sign}(f_1,\dots,f_r,b)$, which will allow us to find the condition on the point $x_n$. A similar occurrence happens when the univariate polynomials are always positive or always negative, as $\operatorname{sign}(f_1,\dots,f_r,b)$ confirms what sign we require. 
\end{proof}

\begin{theorem}
\label{theo:cad}
There is a computable procedure that takes as input a finite list of $n$-variate real polynomials and outputs a finite list $(I_0,\ldots,I_\ell)$ of $\mathrm{AoUC}_{\uint^n}$-instances such that $$\{x \in \uint^n \mid \exists i \ I_i = \{x\}\}$$ is a representative sample for the polynomials.
\end{theorem}

\begin{proof}
The base case is trivial. The unique point $0 \in \mathbb{R}^0$ forms a representative sample for any collection of zero-variate polynomials. We can just output an $\aouc$-instance $\{0\} \in \mathcal{A}(\mathbb{R}^0)$.

Given a finite list of $n+1$-variate real polynomials, we can compute a finite list of $n$-variate polynomials including their projections using Corollary \ref{corr:proj}. By the induction hypothesis, we can compute finitely many $\mathrm{AoUC}_{\uint^n}$-instances such that the determined outputs form a representative sample for the projections. 

Following Lemma \ref{lem:rep} we then obtain the $\mathrm{AoUC}_{\uint^{n+1}}$-instances describing a representative sample for the original polynomials by monitoring the $\mathrm{AoUC}_{\uint^n}$-instances obtained from the projections. Whenever one them specifies a point, we can compute this point, substitute it into the original polynomials and then construct $\mathrm{AoUC}_{\uint^{n+1}}$-instances adding as final component the roots and intermediate values for the resulting univariate polynomials.

By considering the upper bounds on the ranks available to us, we can obtain some upper bound on the number of $\mathrm{AoUC}_{\uint^{n+1}}$-instances required in advance.

%We can substitute in all of those $\mathrm{AoUC}_{\uint^n}$-instances into the original $n+1$-variate polynomials to obtain a finite list of potential univariate polynomials. By Lemma \ref{lem:Aouc}, we can obtain $\aouc$-instances for roots of the polynomials, and add $\aouc$-instances for the midpoints of the intervals between roots. We can then combine the $\mathrm{AoUC}_{\uint^n}$ and the $\aouc$-instances to obtain the desired $\mathrm{AoUC}_{\uint^{n+1}}$-instances as per Lemma \ref{lem:rep}.
\end{proof}

\begin{corollary}
\label{corr:cadmain}
$\aouc^* \leqW \BMRoot \leqW \mathrm{BPIneq} \leqW \aouc^\diamond$.
\end{corollary}

\begin{proof}
The first reduction is a consequence of Proposition \ref{prop:bmrootidempotent}.

Asking for a root of a polynomial $P$ is the same as asking for a solution of $P(\overline{x}) \geq 0 \wedge -P(\overline{x}) \geq 0$; this shows the second reduction.

For the third, we observe that if there is any solution to $\bigwedge_{i \leq k} P_i(\overline{x}) \geq 0$ within $\uint^n$, then every representative sample for $P_0,P_1,\ldots,P_k$ contains a solution. By Theorem~\ref{theo:cad}, $\aouc^*$ lets us obtain a representative sample. Non-solutions will eventually be recognized as such, which is why $\bigsqcup_{n \in \mathbb{N}} \C_n$ can identify a correct solution from finitely many candidates. As shown in \cite{pauly-kihara2-mfcs}, it holds that $\left ( \bigsqcup_{n \in \mathbb{N}} \C_n \right) \star \aouc^* \equivW \aouc^\diamond$.
\end{proof}

We are now prepared to prove our upper bound for the Weihrauch degree of finding Nash equilibria in multiplayer games:

\begin{corollary}
\label{corr:nashaoucdiamond}
$\Nash \leqW \aouc^\diamond$.
\begin{proof}
A strategy profile is a set of strategies for all players which fully specify all actions in a game. We say that a strategy profile is supported on a set $S$ of actions if every action outside of $S$ has probability $0$ in the strategy profile, and for every player the expected payoff for actions in $S$ is at least as much as for every other action. A strategy profile is a Nash equilibrium iff it is supported on some set $S$.

For a fixed set $S$ (for which there are only finitely many candidates), the property of being a strategy profile on it can be expressed as a multivariate polynomial system of inequalities. By compactness, we can detect if such a system has no solution in $\uint^n$. This means that $\left ( \bigsqcup_{n \in \mathbb{N}} \C_n \right)$ can be used to select a set $S$ with the property that some strategy profile is supported on it. We know from Nash's theorem that such a set $S$ must exist.

Once we have selected a suitable $S$, we invoke $\mathrm{BPIneq}$ to actually get a strategy profile supported on it. This is a Nash equilibrium, as desired. We thus get (taking into account Corollary \ref{corr:cadmain}): $$\Nash \leqW \mathrm{BPIneq} \star \left ( \bigsqcup_{n \in \mathbb{N}} \C_n \right) \leqW \aouc^\diamond$$
\end{proof}
\end{corollary}

\subsection{Differences Between our Algorithm and the Original}

We defined sections and sectors in Definition \ref{def:sec} as tuples of continuous real-valued function on the region $\mathcal{R}$. Therefore, the stack decomposition which consists of sections and sectors (Definition \ref{def:decomp}) is also constructed from continuous functions. We speak of an \emph{algebraic} stack decomposition, if these continuous functions are actually all polynomials.

In the original CAD algorithm the decomposition (Definition \ref{def:decomp}) has disjoint regions. Our algorithm, however, can get multiple copies of the same regions. 

\begin{definition}
    Let $\mathcal{R} \subseteq \mathbb{R}^n$. A \emph{weak} decomposition of $\mathcal{R}$ is a finite collection of regions $(R_i)_{i \in I}$ whose union is $\mathcal{R}$ subject to $\mathcal{R}_i \cap \mathcal{R}_j = \emptyset$  or $\mathcal{R}_i = \mathcal{R}_j$ for all $i,j \in I$.

    % Algebraic stack decomp
    %An algebraic stack decomposition of $\R^n$ is a finite collection of polynomials where each polynomial is a stack over $\R$ and $\mathcal{D}$ is an algebraic stack decomposition. 
    
    % Base case $\R^0$
    A weak algebraic stack decomposition of $\R^0$ is just a weak decomposition of $\R^0$. A weak algebraic stack decomposition of $\R^{n+1}$ is a weak decomposition of the form $\bigcup_{\mathcal{R}\in\mathcal{D}}S_{\mathcal{R}}$, where each $S_\mathcal{R}$ is an algebraic stack over $\mathcal{R}$ and $\mathcal{D}$ is a weak algebraic stack decomposition of $\R^n$.
\end{definition}

\begin{definition}
    A (weak) algebraic stack decomposition is minimal for a given set of polynomials $\F$ if each $p \in \F$ is delineable on it, but removing any polynomial from the (weak) algebraic stack decomposition breaks this property.
\end{definition}

If we could test for equality, CAD could produce a minimal algebraic stack decomposition. However, we cannot do this as we do not know the number of pieces/degrees/roots of the polynomials.

\begin{example}
    Consider $\F = \{x^2 + \epsilon\}$ for some parameter $\epsilon \in \R$. If $\epsilon$ is positive, a minimal algebraic stack decomposition for $\F$ uses $0$ polynomials, if $\epsilon = 0$ we need $1$, and if $\epsilon$ is negative we need $2$. Therefore, $\lpo$ reduces to finding minimal algebraic stack decompositions already in the simplest case.

    An algebraic stack decomposition (not weak) for $\mathcal{F}$ consists of a finite tuple of real numbers $(x_1,\dots,x_k)$ that are promised to be distinct and contain the roots of $x^2 + \epsilon$. This allows us to check if $\epsilon$ is zero or negative: if $\epsilon < \frac{1}{2}\min_{i,j \leq k \ i \neq j} |x_i - x_j|^2$, it already has to be the case that $\epsilon = 0$. So even just asking for a algebraic stack decomposition of a single univariate polynomial requires solving $\lpo$.
\end{example}

\begin{theorem}
    Given a finite set of multivariate polynomials, we can compute a weak algebraic stack decomposition such that the original polynomials are delineable over the stack.
\end{theorem}

\begin{proof}
    The basic idea of the CAD algorithm is that for a finite collection $\mathcal{F}$ of $n$-variate polynomials, we obtain a weak algebraic stack decomposition as $\proj^{n} \mathcal{F}, \proj^{n-1} \mathcal{F}, \dots, \proj \mathcal{F}$. Any finite overapproximation of the projections still works, thus Lemma \ref{lem:comp} yields the claim.
\end{proof}

\section{An open question and a remark}
Our main theorem demonstrates that $\aouc^* \leqW \mathrm{Nash} \leqW \aouc^\diamond$ which immediately raises the question of which of those reductions are strict. As previously mentioned, according to \cite{pauly-kihara2-mfcs}  it is established that $\aouc^* \leW \aouc^\diamond$, implying that at least one of the two reductions is strict. Since the degrees of $\aouc^*$ and $\aouc^\diamond$ exhibit significant similarities in various aspects, only a few of the established techniques are available to resolve this situation. While the use of the recursion theorem as demonstrated in \cite{pauly-kihara2-mfcs,pauly-kihara5-tamc} might be possible, it certainly presents a considerable challenge.

A similar question was left unresolved in \cite[Section 5]{pauly-kihara2-mfcs}. In that work, two variants of Gaussian elimination were defined as follows:

\begin{definition}[\cite{pauly-kihara2-mfcs}]
$\textrm{LU-Decomp}_{P,Q}$ takes as input a matrix $A$, and outputs permutation matrices $P$, $Q$, a matrix $U$ in upper echelon form and a matrix $L$ in lower echelon form with all diagonal elements being $1$ such that $PAQ = LU$. By $\textrm{LU-Decomp}_{Q}$ we denote the extension where $P$ is required to be the identity matrix.
\end{definition}

While $\textrm{LU-Decomp}_{P,Q} \equivW \aouc^*$ was demonstrated, for the other variant only $\aouc^* \leqW \textrm{LU-Decomp}_{Q} \leqW \aouc^\diamond$ could be established. A clearer understanding of situations where sequential uses of $\aouc$ are genuinely required to perform some ``algorithm'', and ideally a mathematical theorem or simpler problem which is equivalent to $\aouc^\diamond$ both seem to be very desirable.

Should it hold that $\aouc^* \leW \Nash$, it would be very interesting to see how many players are needed to render the Weihrauch degree of finding Nash equilibria harder than the two-player case. A natural conjecture would be that this already occurs for three players. An important distinction between two-player and three-player games is that two-player games with rational payoffs have rational Nash equilibria, while for every algebraic number $\alpha \in \uint$ there is a three-player game where every Nash equilibrium assigns $\alpha$ as a weight to a particular action, as shown by Bubelis \cite{bubelis}. However, the construction employed by Bubelis does not yield a reduction $\BRoot \leqW \Nash_3$, as it requires a polynomial with $\alpha$ as a simple root as starting point. On the other hand, we know that even $\BRoot \leqW \Nash_2$ holds via Corollary \ref{corr:broot}, albeit with a very roundabout construction. Thus, this particular difference between two and three-player games is immaterial to the Weihrauch degrees concerned.

\section{Consequences of the Classification}
\label{sec:consequences}
In this section, we shall explore some consequences of our classification of the Weihrauch degree of finding Nash equilibria. For this, we consider more permissive notions of algorithms and whether or not they are sufficiently powerful to solve the task. The first important point, however, is that the non-computability of finding Nash equilibria is inherently tied to the potential of having multiple Nash equilibria.

\begin{corollary}
\label{corr:unique}
Let $f : \mathbf{X} \to \mathbf{Y}$ be a function where $\mathbf{Y}$ is computably admissible. Then if $f \leqW \Nash$, then $f$ is already computable.
\begin{proof}
By combining Theorem \ref{theo:main} with \cite[Theorem 2.1]{paulyleroux} (originally \cite[Theorem 5.1]{paulybrattka}), since $\aouc^\diamond \leqW \C_\Cantor$.
\end{proof}
\end{corollary}

An immediate consequence of Corollary \ref{corr:unique} is that if we restrict our consideration to games having a unique Nash equilibrium, then computing the Nash equilibrium is possible. However, this insight goes even further. For example, we could consider the class of games where Player 1 receives the same payoff in any Nash equilibrium. Then computing the equilibrium payoff for Player 1 is possible, even if we might be unable to compute a Nash equilibrium.

As our first extended notion of algorithm, we consider computation with finitely many mindchanges. We begin with a model of computation where the machine continues to output more and more digits of the infinite code for the desired output. We then add the ability for the machine to completely erase all digits written so far, and to start over. To ensure that there is a well-defined output, this ability may be invoked only finitely many times. It was shown in \cite{paulybrattka,paulybrattka2} that a problem $f$ is solvable with finitely many mindchanges iff $f \leqW \C_\mathbb{N}$.

\begin{corollary}
$\Nash$ is solvable with finitely many mindchanges.
\end{corollary}

We can delve a bit deeper and obtain an upper bound for the number of mindchanges required from the dimensions of the game.

Next, we consider various probabilistic models of computation. A Las Vegas machine can use random coin flips to help with its computation. At any point during the computation, it can report a fault and abort, but if it continues running forever, it needs to produce a valid output. For each input, the probability (based on the coin flips) of outputting a correct output needs to be positive (but we do not demand a global positive lower bound). This model was introduced in \cite{hoelzl}. Since Las Vegas computability is closed under composition, we obtain the following strengthening to their \cite[Corollary 17.3]{hoelzl} (by using their \cite[Corollary 16.4]{hoelzl}):

\begin{corollary}
$\Nash$ is Las Vegas computable.
\end{corollary}

It was previously demonstrated in \cite[Theorem 16.6]{hoelzl} that even for a Las Vegas computation solving just $\aouc$ it is not possible to compute a positive lower bound for the success chance from the input -- so in particular, there is no global lower bound.

By dropping the requirement that a wrong guess must be reported at some stage of the computation, we arrive at Monte Carlo machines. They, too, make random coin tosses and are subject to the requirement that any completed output must be correct and that a correct output needs to be given with some positive probability. However, they can fail by simply stopping to produce output\footnote{Due to the Halting problem, we cannot detect whether they have done that.}. This model was introduced in \cite{hoelzl2}, and from the characterizations obtained there together with our classification it follows that:

\begin{corollary}
$\Nash$ is Monte Carlo computable, and moreover, we can compute a positive lower bound for the success chance from the dimensions of the game.
\end{corollary}

Indeed, in the context of probabilistic computation for finding Nash equilibria, we are confronted with a choice. We can opt for either possessing knowledge of a lower bound on the success probability or being able to detect when a random guess during the computation proves to be incorrect. This choice highlights a trade-off between the two aspects within the confines of probabilistic algorithms aimed at solving this problem.
\bibliographystyle{eptcs}
\bibliography{references}

\end{document}